\RequirePackage{fix-cm}

\documentclass[11pt]{svjour3}
\usepackage{amssymb}
\usepackage{amsmath}
\usepackage{mathrsfs}
\usepackage[11pt]{extsizes}
\smartqed  
\usepackage{booktabs,caption}
\usepackage{graphicx}
\usepackage{pdflscape}
\usepackage{geometry}
\usepackage{amsfonts}
\usepackage{xcolor}
\usepackage{longtable}

 \journalname{Quantum Inf. Process.}
\begin{document}
\title{$\mathbb{F}_q\mathcal{R}$-skew cyclic codes and their application to quantum codes}

\titlerunning{$\mathbb{F}_q\mathcal{R}$-skew cyclic codes and their application to quantum codes}        

\author{Om Prakash$^{* 1}$ \and  Shikha Patel$^{1}$ \and Habibul Islam$^{2}$  
}
\authorrunning{Om Prakash \and  Shikha Patel \and Habibul Islam }

\institute{\at
              $^{1}$Department of Mathematics,
              Indian Institute of Technology Patna, Patna 801 106, India \\
              $^{2}$School of Computer Science,          University of St. Gallen, Switzerland \\
              \email{om@iitp.ac.in(*corresponding author), shikha\_1821ma05@iitp.ac.in, habibul.pma17@iitp.ac.in}
                          }
\date{Received: date / Accepted: date}
\maketitle	

\begin{abstract}
Let $p$ be a prime and $\mathbb{F}_q$ be the finite field of order $q=p^m$. In this paper, we study $\mathbb{F}_q\mathcal{R}$-skew cyclic codes where $\mathcal{R}=\mathbb{F}_q+u\mathbb{F}_q$ with $u^2=u$. To characterize $\mathbb{F}_q\mathcal{R}$-skew cyclic codes, we first establish their algebraic structure and then discuss the dual-containing properties by considering a non-degenerate inner product. Further, we define a Gray map over $\mathbb{F}_q\mathcal{R}$ and obtain their $\mathbb{F}_q$-Gray images. As an application, we apply the CSS (Calderbank-Shor-Steane) construction on Gray images of dual containing $\mathbb{F}_q\mathcal{R}$-skew cyclic codes and obtain many quantum codes with better parameters than the best-known codes available in the literature.\\
\end{abstract}
\section{Introduction}
The family of skew cyclic codes is an emerging and interesting class of linear codes introduced by Boucher et al. \cite{Boucher1} in $2007$. In $2009$, they have obtained some new linear codes with record-breaking parameters from skew cyclic codes \cite{Boucher2}. Remarkably, these works have considered the length $n$ of the code as an integral multiple of the order of underlying automorphism. In $2011$, Siap et al. \cite{Siap} relaxed this restriction from $n$ to investigate skew cyclic codes of any length over a finite field. In addition to quasi-cyclic (QC) and quasi-twisted (QT) codes over commutative rings \cite{Ak,Pandey22}, skew linear codes such as skew cyclic, skew QC and skew QT codes produce some excellent parameters \cite{Boucher1,Boucher2}. In fact, in the last two decades, skew cyclic codes produced some new parameters. Therefore, exploring skew cyclic codes remains one of the active research areas. Currently, many researchers have been working on these codes over different alphabets, particularly to produce new codes over finite rings, we refer \cite{Gursoy,P1,Siap}. Also, the number of skew cyclic codes is directly proportional to the number of factors of $x^n-1$. Since the skew polynomial ring is a non-UFD, it  generally possesses more factorization of a polynomial in comparison to a commutative ring. Hence, it is one of the reasons to study such noncommutative skew codes.

On the other hand, many researchers have also worked on codes over mixed alphabets for the last two decades. Linear codes over mixed alphabets were introduced by Brouwer et al. \cite{Bro} in $1998$. In $2010$, Borges et al. \cite{Bor1} studied $\mathbb{Z}_2\mathbb{Z}_4$-linear codes and discussed their generator matrices and parity-check matrices. In $2014$, Abualrub et al. \cite{A} investigated $\mathbb{Z}_2\mathbb{Z}_4$-cyclic codes and proved that the dual of a $\mathbb{Z}_2\mathbb{Z}_4$-cyclic code is again a $\mathbb{Z}_2\mathbb{Z}_4$-cyclic code. Further, in $2015$ Aydogdu et al. \cite{Aydo4} generalized these additive codes to $\mathbb{Z}_{p^r}\mathbb{Z}_{p^s}$. Again, in $2018$ Aydogdu et al. \cite{Aydo3} generalized $\mathbb{Z}_2\mathbb{Z}_4$-cyclic codes to $\mathbb{Z}_2\mathbb{Z}_4\mathbb{Z}_8$-cyclic codes. For more related works on mixed alphabets, we refer \cite{Aydo1,Aydo2,Aydo5,Bor2,Islam.amca,Prakash.amc,Wu}. Recently, in $2020$, Benbelkacem et al. \cite{MN} studied mixed alphabet skew cyclic codes. Here, we discuss structural properties of skew cyclic codes of length $(\alpha,\beta)$
 over the mixed alphabet $\mathbb{F}_q\mathcal{R}$ where $\mathcal{R}=\mathbb{F}_q+u\mathbb{F}_q$ with $u^2=u$. In particular, if $\alpha=0$, we get skew cyclic codes over $\mathcal{R}$ discussed by Gursoy et al. \cite{Gursoy}, and if $\beta=0$, we get skew cyclic codes over $\mathbb{F}_q$. Therefore, we generalize both codes over $\mathbb{F}_q$ and $\mathcal{R}$, which is one of our motivations for this study.

 On the other hand, quantum information theory is a fascinating area of research. It introduces concepts from computer science, classical information theory, and quantum mechanics all at once. Results and techniques of various branches of mathematics, mathematical physics, and quantum statistical physics find applications in this fast-growing field. Classical information theory deals with information-processing functions like storage and transmission of information, whereas quantum information theory studies how these tools can be accomplished using quantum mechanical systems. In classical information theory, error correction occurs in bits; however, in quantum information theory, error correction occurs in quantum states (qubits). Both are concerned with one of the fundamental problems of communication namely reliable transmission of information over a noisy channel. The implementation of quantum error-correcting codes is widely acknowledged as necessary for the development of large-scale, viable quantum computers and the usage of quantum communication. Quantum error-correcting codes (QECCs) have been used to protect quantum information from errors caused by decoherence and other types of quantum noise.
Classical computers use a large number of electrons, so when one goes wrong, it is not too severe. However, in quantum computers single qubit
will probably be just one or a small number of particles. This single qubit already creates a need for some sort of error correction. Therefore, to build reliable quantum computers, quantum error correction will be necessarily needed, while classical computers do not use error correction \cite{Gottesman}.
The construction of quantum error-correcting codes from classical error-correcting codes was independently invented by Shor \cite{Shor} and Steane \cite{Steane}. In 1998, Calderbank et al. \cite{Calderbank}
proposed a systematic method to obtain quantum codes from classical codes, known as
CSS (Calderbank-Shor-Steane) construction.
Under this construction, several investigations have been carried out, for  instance, quantum codes from cyclic codes over
$\mathbb{F}_4+u\mathbb{F}_4,u^2=0$ are found in \cite{Kai11}. In the literature \cite{Dertli15,Dertli15a,Islam18b,Islam.ccds,Qian,Sari}, binary and non-binary
quantum codes from cyclic codes over finite non-chain rings are
well studied. For a detailed study of quantum codes from linear codes, we refer \cite{Alahmadi,Edel,Grassl04,Islam19,Islam20b,Islam20a,Shi20}.
The number of (dual containing) skew codes depends on the factorization of $x^n-1$, and hence the major reason behind the investigation of skew cyclic codes in the construction of quantum code lies in the factorization of $x^n- 1$, which sometimes has more choices for $h(x)$ in a skew polynomial ring. We note that quantum codes from skew codes have appeared in very few papers \cite{bag,Prakash2,Prakash,Verma21}.

Further, we notice that in the case of mixed alphabets, researchers mainly focused on exploring the structural properties such as generator matrices, parity check matrices, generating polynomials, minimal generating sets, generating polynomials for dual codes, etc. Presently, only a few works are available in the literature on the applications of mixed alphabet codes \cite{Islam20b,Li20}. Recently, Li et al. \cite{Li20} studied $\mathbb{F}_qR$-linear skew constacyclic codes and constructed quantum codes by the Hermitian construction. In continuation with that framework, here we find quantum codes in the case of  noncommutative mixed alphabets codes. Two significant contributions of the paper are given below.
\begin{enumerate}
    \item We discuss the algebraic structure of $\mathbb{F}_q\mathcal{R}$-skew cyclic codes and their dual codes. Then a necessary and sufficient condition of these codes to contain their dual is provided.
    \item As an application, we obtain quantum codes with better parameters (compared to the recent literature) from dual containing codes.
\end{enumerate}	
The paper is structured as follows: In Section $2$, we present some basic results and definitions followed by the algebraic structure of the ring $\mathbb{F}_q\mathcal{R}$. Section $3$ presents the algebraic properties of $\mathbb{F}_q\mathcal{R}$-skew cyclic codes. We have shown that the dual of an $\mathbb{F}_q\mathcal{R}$-skew cyclic code is again an $\mathbb{F}_q\mathcal{R}$-skew cyclic code. Section $4$ includes the Gray map through which we discuss images of $\mathbb{F}_q\mathcal{R}$-skew cyclic codes. Moreover, we have provided some examples of $\mathbb{F}_q\mathcal{R}$-skew cyclic codes and their Gray images over $\mathbb{F}_q$. Section $5$ gives a method to construct quantum codes. Section $6$ concludes the paper.
	\section{Preliminary}
	 Let $\mathbb{F}_q$ be a finite field and  $q=p^m$ for some prime $p$ and a positive integer $m$. A $k$-dimensional subspace $\mathcal{C}$ of $\mathbb{F}_q^n$ is said to be a linear code of length $n$ over $\mathbb{F}_q$ and every element $c=(c_0,c_1,\dots,c_{n-1})\in \mathcal{C}$ is called a codeword. Let $c\in\mathbb{F}_q^n$. The Hamming weight of $c$, denoted by $w_H(c)$, is defined as the number of nonzero coordinates in $c$. For any two vectors $c$ and $c'$ in  $\mathbb{F}_q^n$, the
	Hamming distance between $c$ and $c'$, denoted by $d_H(c,c')$, is defined as the number
	of places in which $c$ and $c'$ differ.
 For a linear code $\mathcal{C}$, the Hamming distance is given by $$d_H(\mathcal{C})=\min\{d_H(c,c')~|~ c\neq c', \text{for all}~ c,c'\in \mathcal{C} \}.$$
Let $c=(c_0,c_1,\dots,c_{n-1})$ and $c'=(c'_0,c'_1,\dots,c'_{n-1})$	be two vectors. The Euclidean inner product of $c$ and $c'$ in  $\mathbb{F}_q^n$  is defined by  $c\cdot c' = \sum_{i=0}^{n-1}c_ic'_i$. The dual code of $\mathcal{C}$ is defined by  $\mathcal{C}^\perp= \{c\in \mathbb{F}_q^n~|~ c\cdot c' = 0, ~\text{for all} ~c'\in \mathcal{C} \}$. A linear code $\mathcal{C}$ is self-orthogonal if $\mathcal{C}\subseteq \mathcal{C}^\perp$ and self-dual if $\mathcal{C}=\mathcal{C}^\perp$. The parameters of a code with minimum distance $d$ and dimension $k$ are written as $[n, k, d]$. If $\mathcal{C}$ is an $[n, k, d]$ code, then from
the Singleton bound, its minimum distance is bounded above by $$d \leq n - k + 1.$$
A code meeting the above bound  with equality is called maximum-distance-separable (MDS). We call a code is
almost MDS if its minimum distance is one unit less than the MDS bound. A code is called optimal if it has the highest possible minimum distance for a given length and dimension. \\
Throughout this paper, $\mathcal{R}$ represents the quotient ring $\mathbb{F}_q[u]/\langle u^2-u \rangle$. It can be easily seen that $\mathcal{R}$ is a finite commutative non-chain ring containing $q^2$ elements with characteristic $p$. It is a semi-local ring
with maximal ideals $\langle u \rangle$  and $\langle 1-u \rangle$. Let $\xi_1 = 1-u,
\xi_2 = u$. Then $\xi_i^2=\xi_i$,  $\xi_i\xi_j=0$ and $\sum_{k=1}^{2}\xi_k=1$ where $i, j = 1, 2$ and $i\neq j$.Thus, we have
$$ \mathcal{R}= \xi_1\mathcal{R} \oplus \xi_2\mathcal{R}.$$ Further, $\xi_i\mathcal{R} \cong \xi_i\mathbb{F}_q$, $i = 1, 2$. Therefore, any $r \in \mathcal{R}$ can be uniquely written as $r =\sum_{i=1}^{2}\xi_ia_i$ where $a_i\in \mathbb{F}_q$ for $i = 1, 2$. A linear code $\mathcal{C}$ of length $n$ over $\mathcal{R}$ is an $\mathcal{R}$-submodule of $\mathcal{R}^n$. Let $B_i$ be a linear code over $\mathcal{R}$, for $i=1,2$. Then the operations $\oplus$ and $\otimes$ are defined as
$$B_1\oplus B_2=\{b_1+b_2~|~ b_i\in B_i~ \text{for all}~i\}$$ and $$B_1\otimes B_2=\{(b_1,b_2)~|~ b_i\in B_i~ \text{for all}~i\},$$  respectively. We know that any element of the ring $\mathcal{R}$ can be uniquely written as $$a+ub=(a+b)u + a(1-u)$$ where $a, b\in \mathbb{F}_q$. Therefore, $a + ub \in \mathcal{R}$ is a unit if and only if $a, a+b\in \mathbb{F}_q^*$. Now, following \cite{Gursoy}, let $\mathcal{C}$ be a linear code of length $n$ over $\mathcal{R}$ and
\begin{align*}
\mathcal{C}_1&=\{\bold{a}\in  \mathbb{F}_q^n~ |~\bold{a}+u\bold{b}\in \mathcal{C},~ \text{for some}~ \bold{b}\in  \mathbb{F}_q^n   \};\\
\mathcal{C}_2&=\{\bold{a}+\bold{b}\in  \mathbb{F}_q^n~ |~\bold{a}+u\bold{b}\in \mathcal{C}  \}
\end{align*}
Then $\mathcal{C}_1$ and $\mathcal{C}_2$ are linear codes of length $n$ over $\mathbb{F}_q$ and $\mathcal{C}$ can be uniquely written as $\mathcal{C}= \xi_1\mathcal{C}_1 \oplus \xi_2\mathcal{C}_2$. Let $s=a+ub\in \mathcal{R}$ and $c=(c_0,c_1,\dots,c_{n-1})\in \mathcal{C}$ where $c_i=a_i+ub_i$ with $a_i,b_i\in \mathbb{F}_q$ for $0\leq i\leq n-1$. The $i^{th}$-entry of $sc$ is given as
\begin{align*}
(a+ub)(a_i+ub_i)&=((a+b)u-a(u-1))(a_i+ub_i)\\
&=aa_i+u(ba_i+(b+a)b_i)\\
&=(1-u)r+u(r+t)
\end{align*}
where $r=aa_i$ and $t=ba_i+bb_i+ab_i$. Hence, $sc$ can be written in terms of $\mathcal{C}_1$ and $\mathcal{C}_2$ with $r=a(a_0,a_1,\dots,a_{n-1})$ and $t=(a+b)(b_0,b_1,\dots,b_{n-1})+b(a_0,a_1,\dots,a_{n-1})$.\\
Now, we introduce two classes of noncommutative rings $\mathbb{F}_q[x;\Theta]$ and $\mathcal{R}[x;\theta]$.
\begin{definition}
Let $\Theta$ be an automorphism of $\mathbb{F}_q$ defined by $\Theta(a)=a^{p^i}$ and $q=p^m$. We consider the set
$$\mathbb{F}_q[x;\Theta]=\{b_ex^e+\cdots+b_1x+b_0~|~ b_i\in \mathbb{F}_q~ \text{and}~ e\in \mathbb{N}\}.$$ Then $\mathbb{F}_q[x;\Theta]$ is a noncommutative ring unless $\Theta$ is the identity under the usual addition of polynomials and multiplication defined by $(ax^i)(bx^j)=a\Theta^i(b)x^{i+j}$ for $a, b\in \mathbb{F}_q$ and associative and distributive rules over $\mathbb{F}_q[x;\Theta]$.
\end{definition}
Now, consider $$\mathcal{R}[x;\theta]=\{a_lx^l+\cdots+a_1x+a_0~|~ a_i\in \mathcal{R}~ \text{and}~ l\in \mathbb{N}\}$$ where $\theta$ is the automorphism of $\mathcal{R}$ defined by $\theta(a+ub)=a^{p^{i}}+ub^{p^{i}}$. Hence the order of the automorphism is $|\langle \theta \rangle|=\frac{m}{\gcd {(m,i)}}$ and in particular,  $|\langle \theta \rangle|=\frac{m}{i}$ when $i|m$. Also, the subring of $\mathbb{F}_{p^i}+u\mathbb{F}_{p^i}$  is invariant under $\theta$. Then $\mathcal{R}[x;\theta]$ is a noncommutative ring unless $\theta$ is the identity under the usual addition of polynomials and multiplication defined by  $(sx^i)(s'x^j)=s\theta^i(s')x^{i+j}$ for $s,s'\in \mathcal{R}$ and associative and distributive rules over $\mathcal{R}[x;\theta]$. This ring is known as a skew polynomial ring.
 \begin{definition}
 A linear code $\mathscr{C}$ of length $\beta$ over $\mathcal{R}$ is said to be a skew cyclic code with respect to the automorphism $\theta$ if for any codeword $c=(c_0,c_1,...,c_{\beta-1})\in \mathscr{C}$ we have $\delta(c)=(\theta(c_{\beta-1}),\theta({c_0}),...,\theta(c_{\beta-2})) \in \mathscr{C}$ where $\delta$ is a skew cyclic shift of $\mathscr{C}$.
 \end{definition}
Here, $\frac{\mathcal{R}[x;\theta]}{\langle x^\beta-1 \rangle}$ forms a ring when $|\theta|$ divide $\beta$. But for $|\theta|\nmid \beta$, $\frac{\mathcal{R}[x;\theta]}{\langle x^\beta-1 \rangle}$ fails to be a ring. However, in that case it is a
left $\mathcal{R}[x;\theta]$-module under the left multiplication defined by $f(x)\big(g(x)+\langle x^\beta-1\rangle\big)=f(x)g(x)+\langle x^\beta-1\rangle$
where $f(x),g(x)\in \mathcal{R}[x;\theta]$. As usual, it is convenient to identify each codeword $c= (c_0,c_1,...,c_{\beta-1})\in \mathscr{C}\subseteq \mathcal{R}^\beta$ with the polynomial $c(x)=c_0+c_1x+\cdots+c_{\beta-1}x^{\beta-1}\in \frac{\mathcal{R}[x;\theta]}{\langle x^\beta-1 \rangle}$ under the correspondence $c=(c_0,c_1,...,c_{\beta-1})\longrightarrow c(x)=c_0+c_1x+\cdots+c_{\beta-1}x^{\beta-1}$.  As a consequence, every code
of length $\beta$ over $\mathcal{R}$ can be viewed as a subset of $\frac{\mathcal{R}[x;\theta]}{\langle x^\beta-1 \rangle}$.  The following theorems are used to derive the results in further sections.
 \begin{theorem}\cite[Theorem 3]{Gursoy}
 Let $\mathscr{C}= \xi_1\mathcal{C}_1 \oplus \xi_2\mathcal{C}_2$ be a linear code of length $\beta$ over $\mathcal{R}$. Then $\mathscr{C}$ is a skew cyclic code of length $\beta$ over $\mathcal{R}$ if and only if $\mathcal{C}_1$ and $\mathcal{C}_2$ are skew cyclic codes of length $\beta$ over $\mathbb{F}_q$.
 \end{theorem}
 \begin{theorem}\cite[Theorem 5]{Gursoy} \label{th2}
 Let $\mathscr{C}= \xi_1\mathcal{C}_1\oplus \xi_2\mathcal{C}_2$ be a skew cyclic code of length $\beta$ over $\mathcal{R}$. Let $g_1(x)$ and $g_2(x)$ be generator polynomials of $\mathcal{C}_1$ and $\mathcal{C}_2$, respectively. Then $\mathscr{C}=\langle
 \xi_1g_1(x)+\xi_2g_2(x)\rangle$.
 \end{theorem}
 Let $s=a+ub$ be an element of $\mathcal{R}$. Then we define a map $\eta: \mathcal{R} \rightarrow \mathbb{F}_q$ as follows
 $$\eta(s)=a.$$ It is clear that $\eta$ is a ring homomorphism. Let $\mathbb{F}_q\mathcal{R} =\{(x,y)~|~x\in \mathbb{F}_q~ \text{and}~ y\in\mathcal{R}\}$. Now for any $s\in \mathcal{R}$ and $(x,y)\in \mathbb{F}_q\mathcal{R}$, we define an $\mathcal{R}$-scalar multiplication on $\mathbb{F}_q\mathcal{R}$ as $$*: \mathcal{R}\times \mathbb{F}_q\mathcal{R} \longrightarrow \mathbb{F}_q\mathcal{R}$$ such that
 $$s*(x,y)=(\eta(s)x,sy).$$ It is easy to verify that $\mathbb{F}_q\mathcal{R}$ is an $\mathcal{R}$-module under this multiplication. Further, it can be extended componentwise to $\mathbb{F}_q^\alpha\times\mathcal{R}^\beta$ as follows:
 $$*: \mathcal{R}\times (\mathbb{F}_q^\alpha\times\mathcal{R}^\beta) \longrightarrow \mathbb{F}_q^\alpha\times\mathcal{R}^\beta$$ where
 $$s*l=(\eta(s)x_0,\eta(s)x_1,\dots,\eta(s)x_{\alpha-1},sy_0,sy_1,\dots,sy_{\beta-1});$$ with $s\in \mathcal{R}$ and $l=(x_0,x_1,\dots,x_{\alpha-1},y_0,y_1,\dots,y_{\beta-1})\in\mathbb{F}_q^\alpha\times\mathcal{R}^\beta $. Under this multiplication, $\mathbb{F}_q^\alpha\times\mathcal{R}^\beta$ forms an $\mathcal{R}$-module.
 \begin{definition}
 A nonempty subset $\mathcal{C}$ of $\mathbb{F}_q^\alpha\times\mathcal{R}^\beta$ is said to be an $\mathbb{F}_q\mathcal{R}$-linear code of length $(\alpha,\beta)$ if $\mathcal{C}$ is an $\mathcal{R}$-submodule of $\mathbb{F}_q^\alpha\times\mathcal{R}^\beta$.
 \end{definition}
Now, we define an inner product between $l=(x_0,x_1,\dots,x_{\alpha-1},y_0,y_1,\dots,y_{\beta-1})$ and $l'=(x_0',x_1',\dots,x_{\alpha-1}',y_0',y_1',\dots,y_{\beta-1}')$ as
 $$l\cdot l' =u\sum_{i=0}^{\alpha -1}x_ix_i'+\sum_{j=0}^{\beta -1}y_jy_j'.$$ Let $\mathcal{C}$ be an $\mathbb{F}_q\mathcal{R}$-linear code of length $(\alpha,\beta)$. Then the dual code of $\mathcal{C}$ is defined as
 $$\mathcal{C}^\perp=\{l'\in \mathbb{F}_q^{\alpha}\times\mathcal{R}^{\beta}~|~l\cdot l'=0~ \forall~ l\in \mathcal{C}\}.$$
 Denote $R_{\alpha,\beta}=\frac{\mathbb{F}_q[x;\Theta]}{\langle x^\alpha-1\rangle}\times\frac{\mathcal{R}[x;\theta]}{\langle x^\beta-1\rangle}$, $k(x)=k_0+k_1x+\cdots+k_{\alpha-1}x^{\alpha-1}\in \frac{\mathbb{F}_q[x;\Theta]}{\langle x^\alpha-1\rangle}$ and $t(x)=t_0+t_1x+\cdots+t_{\beta-1}x^{\beta-1}\in \frac{\mathcal{R}[x;\theta]}{\langle x^\beta-1\rangle}$. Then any vector  $c=(k_0,k_1,\dots,k_{\alpha-1},t_0,t_1,\dots,t_{\beta-1})\in \mathbb{F}_q^\alpha\times\mathcal{R}^\beta$ can be identified with a polynomial of the form
 $$c(x)=(k(x),t(x))\in R_{\alpha,\beta}$$ which gives a one-to-one correspondence between $\mathbb{F}_q^\alpha\times\mathcal{R}^\beta$ and $R_{\alpha,\beta}$. Let $r(x)=r_0+r_1x+\cdots+r_{\gamma-1}x^{\gamma-1}\in \mathcal{R}[x;\theta]$ and $(k(x),t(x))\in R_{\alpha,\beta}$. Then the multiplication $*$ in $R_{\alpha,\beta}$ is defined as follows:
 $$r(x)*(k(x),t(x))=(\eta(r(x))k(x),r(x)t(x))$$ where $\eta(r(x))=\eta(r_0)+\eta(r_1)x+\cdots+\eta(r_{\gamma-1})x^{\gamma-1}$. Then under the above defined multiplication $*$, $R_{\alpha,\beta}$ forms a left $\mathcal{R}[x;\theta]$-module.
\section{Structural Properties of $\mathbb{F}_q\mathcal{R}$-Skew Cyclic Codes}
In this section, first, we define $\mathbb{F}_q\mathcal{R}$-skew cyclic codes. Then algebraic properties of these codes are discussed in detail. Also,  the form of a  generator polynomial of an $\mathbb{F}_q\mathcal{R}$-skew cyclic code in $R_{\alpha,\beta}$ is established.
\begin{definition}
An $\mathcal{R}$-submodule $\mathcal{C}$ of $\mathbb{F}_q^\alpha\times\mathcal{R}^\beta$ is said to be an $\mathbb{F}_q\mathcal{R}$-skew cyclic code of length $n=(\alpha,\beta)$ if for any $c=(x_0,x_1,\dots,x_{\alpha-1},y_0,y_1,\dots,y_{\beta-1})\in\mathcal{C}$ we have $\sigma(c):=(\Theta(x_{\alpha-1}),\Theta(x_0),\Theta(x_1),\dots,\Theta(x_{\alpha-2}),\theta(y_{\beta-1}),\theta(y_0),\theta(y_1),\dots,\theta(y_{\beta-2}))\in\mathcal{C}$ where $\sigma$ is the skew cyclic shift of $c$.
\end{definition}
\begin{theorem}\label{thdual}
Let $\mathcal{C}$ be an $\mathbb{F}_q\mathcal{R}$-skew cyclic code of length $n=(\alpha,\beta)$ and $|\langle\theta\rangle|$ divide $\beta$. Then $\mathcal{C}^\perp$ is also an $\mathbb{F}_q\mathcal{R}$-skew cyclic code of length $n=(\alpha,\beta)$.
\end{theorem}
\begin{proof}
Let $\mathcal{C}$ be an $\mathbb{F}_q\mathcal{R}$-skew cyclic code of length $n=(\alpha,\beta)$ and $l'=(x_0',x_1',\dots,x_{\alpha-1}',y_0',y_1',\\\dots,y_{\beta-1}')\in \mathcal{C}^\perp$. Assume that $l=(x_0,x_1,\dots,x_{\alpha-1},y_0,y_1,\dots,y_{\beta-1})\in \mathcal{C}$ and $lcm(\alpha,\beta)=\mathcal{L}$. Now, we have to show that $\sigma(l')\in\mathcal{C}^\perp$. Consider the inner product of $l$ and $\sigma(l')$. We have
\begin{align*}
l\cdot\sigma(l')=&(x_0,x_1,\dots,x_{\alpha-1},y_0,y_1,\dots,y_{\beta-1})\cdot\\& (\Theta(x_{\alpha-1}'),\Theta(x_0'),\Theta(x_1'),\dots,\Theta(x_{\alpha-2}'),\theta(y_{\beta-1}'),\theta(y_0'),\theta(y_1'),\dots,\theta(y_{\beta-2}'))\\
=&u(x_0\Theta(x_{\alpha-1}')+x_1\Theta(x_0')+\cdots+x_{\alpha-1}\Theta(x_{\alpha-2}'))\\&+(y_0\theta(y_{\beta-1}')+y_1\theta(y_0')+\cdots+y_{\beta-1}\theta(y_{\beta-2}')).
\end{align*}
We only need to show that $x_0\Theta(x_{\alpha-1}')+x_1\Theta(x_0')+\cdots+x_{\alpha-1}\Theta(x_{\alpha-2}')=0$ and $y_0\theta(y_{\beta-1}')+y_1\theta(y_0')+\cdots+y_{\beta-1}\theta(y_{\beta-2}')=0$. As $\mathcal{C}$ is an $\mathbb{F}_q\mathcal{R}$-skew cyclic code, $\sigma^{\mathcal{L}-1}(l)\in \mathcal{C}$ where $\sigma^{\mathcal{L}-1}(l)=(\Theta^{\mathcal{L}-1}(x_1),\Theta^{\mathcal{L}-1}(x_2),\dots,\Theta^{\mathcal{L}-1}(x_{\alpha-1}),\Theta^{\mathcal{L}-1}(x_0),\theta^{\mathcal{L}-1}(y_1),\theta^{\mathcal{L}-1}(y_2),\dots,\\\theta^{\mathcal{L}-1}(y_{\beta-1}),\theta^{\mathcal{L}-1}(y_0))$. Now, we get that $\sigma^{\mathcal{L}-1}(l)\cdot l'=0$, where
\begin{align*}
\sigma^{\mathcal{L}-1}(l)\cdot l'&=u\sum_{i=0}^{\alpha -1}\Theta^{\mathcal{L}-1}(x_{i+1})x_{i}'+\sum_{j=0}^{\beta -1}\theta^{\mathcal{L}-1}(y_{j+1})y_{j}'.
\end{align*}
Comparing the coefficients on both sides, we get
$$\Theta^{\mathcal{L}-1}(x_0)x_{\alpha-1}'+\Theta^{\mathcal{L}-1}(x_1)x_0'+\cdots+\Theta^{\mathcal{L}-1}(x_{\alpha-1})x_{\alpha-2}'=0$$ and
$$\theta^{\mathcal{L}-1}(y_0)y_{\beta-1}'+\theta^{\mathcal{L}-1}(y_1)y_0'+\cdots+\theta^{\mathcal{L}-1}(y_{\beta-1})y_{\beta-2}'=0.$$ Since the order of $\theta$ divides $\beta$, $|\langle \theta\rangle|$ divide $lcm(\alpha,\beta)=\mathcal{L}$ and hence $\theta^{\mathcal{L}}(a)=\Theta^{\mathcal{L}}(a)=a$ for any $a\in \mathbb{F}_q$.  Then
$\theta(\theta^{\mathcal{L}-1}(y_0)y_{\beta-1}'+\theta^{\mathcal{L}-1}(y_1)y_0'+\cdots+\theta^{\mathcal{L}-1}(y_{\beta-1})y_{\beta-2}')=\theta(0)=0$ and $\Theta(\Theta^{\mathcal{L}-1}(x_0)x_{\alpha-1}'+\Theta^{\mathcal{L}-1}(x_1)x_0'+\cdots+\Theta^{\mathcal{L}-1}(x_{\alpha-1})x_{\alpha-2}')=\Theta(0)=0$. This implies $y_0\theta(y_{\beta-1}')+y_1\theta(y_0')+\cdots+y_{\beta-1}\theta(y_{\beta-2}')=0$ and $x_0\Theta(x_{\alpha-1}')+x_1\Theta(x_0')+\cdots+x_{\alpha-1}\Theta(x_{\alpha-2}')=0$. Therefore, $l\cdot\sigma(l')=0$ and hence $\sigma(l')\in \mathcal{C}^\perp$. Thus, $\mathcal{C}^\perp$ is an $\mathbb{F}_q\mathcal{R}$-skew cyclic code of length $(\alpha,\beta)$.\qed
\end{proof}
\begin{theorem}
A code $\mathcal{C}$ is an $\mathbb{F}_q\mathcal{R}$-skew cyclic code of length $(\alpha,\beta)$ if and only if $\mathcal{C}$
is a left $\mathcal{R}[x;\theta]$-submodule of $R_{\alpha,\beta}$.
\end{theorem}
\begin{proof}
Let $\mathcal{C}$ be an $\mathbb{F}_q\mathcal{R}$-skew cyclic code of length $(\alpha,\beta)$ and $c=(k_0,k_1,\dots,k_{\alpha-1},t_0,t_1,\dots,\\t_{\beta-1})\in \mathcal{C}$; and the corresponding element of $c$ in $R_{\alpha,\beta}$ be $c(x)=(k(x),t(x))$. Now, consider
\begin{align*}
x*c(x)&=(\Theta(k_{\alpha-1})+\Theta(k_0)x+\cdots+\Theta(k_{\alpha-2})x^{\alpha-1}, \theta(t_{\beta-1})+\theta(t_0)x+\cdots+\theta(t_{\beta-2})x^{\beta-1})
\end{align*} which corresponds to the skew cyclic shift $(\Theta(k_{\alpha-1}),\Theta(k_1),\dots,\Theta(k_{\alpha-2}),\theta(t_{\beta-1}),\theta(t_1),\dots,\\\theta(t_{\beta-2}))$ of $(k_0,k_1,\dots,k_{\alpha-1},t_0,t_1,\dots,t_{\beta-1})$. Therefore, $x*c(x)\in \mathcal{C}$. Then by linearity of $\mathcal{C}$, $s(x)*c(x)\in \mathcal{C}$ for any $s(x)\in \mathcal{R}[x;\theta]$ and hence $\mathcal{C}$
is a left $\mathcal{R}[x;\theta]$-submodule of $R_{\alpha,\beta}$. \\
The converse part follows from the definition.\qed
\end{proof}
Let $\mathcal{C}$ be an $\mathbb{F}_q\mathcal{R}$-skew cyclic code of length $(\alpha,\beta)$. Let $c(x)=(k(x),t(x))\in \mathcal{C}$ and $m(x)\in \frac{\mathbb{F}_q[x;\Theta]}{\langle x^\alpha-1\rangle}$. Now, we consider
$$M=\{t(x)\in \mathcal{R}[x;\theta]/\langle x^\beta-1\rangle ~|~(m(x),t(x))\in \mathcal{C}\}$$ and $$N=\{k(x)\in \mathbb{F}_q[x;\Theta]/\langle x^\alpha-1\rangle ~|~(k(x),0)\in \mathcal{C}\}.$$ Next, we present the algebraic properties of these two sets.
\begin{lemma}\label{lem1}
The above defined set $N$ is a left $\mathbb{F}_q[x;\Theta]$-submodule of $\mathbb{F}_q[x;\Theta]/\langle x^\alpha-1\rangle$ generated by a right divisor of $x^\alpha-1$.
\end{lemma}
\begin{proof}
Let $k_1(x)$ and $k_2(x)$ be two elements of $N$. Then from definition, $(k_1(x),0)\in \mathcal{C}$ and $(k_2(x),0)\in \mathcal{C}.$ Further, $(k_1(x),0)+(k_2(x),0)=(k_1(x)+k_2(x),0)\in \mathcal{C}$. This implies that $k_1(x)+k_2(x)\in N$. Again, let $t'(x)\in\mathbb{F}_q[x;\Theta]/\langle x^\alpha-1\rangle$ and $k(x)\in N$. Then $(k(x),0)\in  \mathcal{C}$. As $\mathcal{C}$ is a left $\mathcal{R}[x;\theta]$-submodule, we have
$$t'(x)*(k(x),0)=(t'(x)k(x),0)\in \mathcal{C}.$$ This implies $t'(x)k(x)\in N$ where $t'(x)k(x)$ is taken modulo $x^{\alpha}-1$. Hence, $N$ is a left $\mathbb{F}_q[x;\Theta]$-submodule of $\mathbb{F}_q[x;\Theta]/\langle x^\alpha-1\rangle$ generated by a right divisor $f(x)$ of $x^\alpha-1$.\qed
\end{proof}
\begin{lemma} \label{lem2}
The set $M$ is a left $\mathcal{R}[x;\theta]$-submodule of $\mathcal{R}[x;\theta]/\langle x^\beta-1\rangle$ which is generated by a single element.
\end{lemma}
\begin{proof}
Let $t_1(x)$ and $t_2(x)$ be two elements of $M$. Then from definition there exist $m_1(x)$ and $m_2(x)$ in $\mathbb{F}_q[x;\Theta]/\langle x^\alpha-1\rangle$ such that $(m_1(x),t_1(x))\in \mathcal{C}$, $(m_2(x),t_2(x))\in \mathcal{C}$. Thus, $$(m_1(x),t_1(x))+(m_2(x),t_2(x))=(m_1(x)+m_2(x),t_1(x)+t_2(x))\in \mathcal{C}.$$ This implies that $t_1(x)+t_2(x)\in M$. Let $s(x)\in \mathcal{R}[x;\theta]/\langle x^\beta-1\rangle$ and $(m(x),t(x))\in \mathcal{C}$. Since $\mathcal{C}$ is a left $\mathcal{R}[x;\theta]$-submodule of $R_{\alpha,\beta}$, we have
$$s(x)*(m(x),t(x))=(\eta(s(x))m(x),s(x)t(x))\in \mathcal{C}.$$ This implies $s(x)t(x)\in M$, where $s(x)t(x)$ is taken modulo $x^{\beta}-1$ and $\eta(s(x))m(x)$ is taken modulo $x^{\alpha}-1$. Therefore, $M$ is a left $\mathcal{R}[x;\theta]$-submodule of $\mathcal{R}[x;\theta]/\langle x^\beta-1\rangle$. Further, from Theorem \ref{th2}, $M=\langle
g(x)\rangle$ where $g(x):=
 \xi_1g_1(x)+\xi_2g_2(x)$.\qed
\end{proof}
\begin{theorem}
Let $\mathcal{C}$ be an $\mathbb{F}_q\mathcal{R}$-skew cyclic code of length $(\alpha,\beta)$. Then $\mathcal{C}$ is generated as a left $\mathcal{R}[x;\theta]$-submodule of $R_{\alpha,\beta}$ by $(f(x),0)$ and $(m(x),g(x))$ where $m(x)\in \mathbb{F}_q[x;\Theta]/\langle x^\alpha-1\rangle$, $g(x)=\xi_1g_1(x)+\xi_2g_2(x)$ and $f(x)$ is a right divisor of $x^{\alpha}-1$.
\end{theorem}
\begin{proof}
Let $c=(c_1,c_2)\in \mathcal{C}$ with $c_1\in \mathbb{F}_q^{\alpha}$ and $c_2\in \mathcal{R}^{\beta}$. It has the polynomial representation as $c(x)=(c_1(x),c_2(x))$ in $R_{\alpha,\beta}$. Then $c_2(x)\in M$ and from Lemma \ref{lem2}, $c_2(x)=h(x)g(x)$ for some $h(x)\in \mathcal{R}[x;\theta]/\langle x^\beta-1\rangle$. As $g(x)\in M$, there exist $m(x)\in \mathbb{F}_q[x;\Theta]/\langle x^\alpha-1\rangle$ such that $(m(x),g(x))\in \mathcal{C}$. Now,
\begin{align*}
c(x)&=(c_1(x),c_2(x))\\
&=(c_1(x),0)+(0,h(x)g(x))\\
&=(c_1(x),0)+p(\eta(h(x)m(x)),0)+(0,h(x)g(x))\\
&=(c_1(x),0)+(p-1)(\eta(h(x))m(x),0)+(\eta(h(x)m(x)),0)+(0,h(x)g(x))\\
&=(c_1(x),0)+(p-1)(\eta(h(x))m(x),0)+(\eta(h(x))m(x),h(x)g(x))\\
&=(c_1(x)+(p-1)\eta(h(x))m(x),0)+(\eta(h(x))m(x),h(x)g(x)).
\end{align*} where $p$ is the characteristic of $\mathbb{F}_q$.
This gives $(c_1(x)+(p-1)\eta(h(x))m(x),0)\in \mathcal{C}$ and hence $(c_1(x)+(p-1)\eta(h(x))m(x),0)\in N$. Moreover, by Lemma \ref{lem1} there exists $d(x)\in N$ such that $c_1(x)+(p-1)\eta(h(x))m(x)=d(x)f(x)$. Therefore, $c(x)=(\eta(h(x))m(x),h(x)g(x))+(d(x)f(x),0)=h(x)*(m(x),g(x))+d(x)*(f(x),0)$.\qed
\end{proof}
\begin{theorem}
Let $\mathcal{C}=\langle(f(x),0),(m(x),g(x))  \rangle$ be an $\mathbb{F}_q\mathcal{R}$-skew cyclic code of length $(\alpha,\beta)$ with $m(x)=0$. Then $\mathcal{C}=\mathcal{C}'\otimes \mathscr{C}$ where $\mathcal{C}'$ is a skew cyclic code of length $\alpha$ over $\mathbb{F}_q$ and $\mathscr{C}$ is a skew cyclic code of length $\beta$ over $\mathcal{R}$.
\end{theorem}
\begin{proof}
Let $c=(c_1,c_2)\in \mathcal{C}$ with $c_1\in \mathbb{F}_q^{\alpha}$ and $c_2\in \mathcal{R}^{\beta}$.
Then, from Lemma \ref{lem1} and Lemma \ref{lem2}, $c_1\in \mathcal{C}'=\langle f(x) \rangle$ and $c_2\in \mathscr{C}=\langle g(x) \rangle$. This implies $\mathcal{C}=\mathcal{C}'\otimes \mathscr{C}$ where $\mathcal{C}'=\langle f(x) \rangle$ and $\mathscr{C}=\langle g(x) \rangle$.\qed
\end{proof}
\begin{corollary}\label{cor1}
Let $\gcd(\alpha, |\langle \Theta\rangle|)=1$ and $\mathcal{C}=\langle(f(x),0),(0,g(x))  \rangle$ be an $\mathbb{F}_q\mathcal{R}$-skew cyclic code of length $(\alpha,\beta)$. Then $\mathcal{C}=\mathcal{C}'\otimes \mathscr{C}$ where $\mathcal{C}'$ is a cyclic code of length $\alpha$ over $\mathbb{F}_q$ and $\mathscr{C}$ is a skew cyclic code of length $\beta$ over $\mathcal{R}$.
\end{corollary}
\begin{proof}
Straightforward.\qed
\end{proof}
\begin{definition}
    Let $\mathcal{C}$ be an $\mathbb{F}_q\mathcal{R}$-linear code of length $(\alpha,\beta)$
 and $\mathcal{C}_\alpha$ (respectively, $\mathcal{C}_\beta$) be
the canonical projection of $\mathcal{C}_\alpha$  on the first $\alpha$ (respectively, on the last $\beta$) coordinates. The
code $\mathcal{C}$  is called separable if $\mathcal{C}$  is the direct product of $\mathcal{C}_\alpha$  and $\mathcal{C}_\beta$, i.e., $\mathcal{C}=\mathcal{C}_\alpha\otimes\mathcal{C}_\beta$.
\end{definition}
 Let $\mathcal{C}'$ be a skew cyclic code over $\mathbb{F}_q$, and $\mathscr{C}$ be a skew cyclic code over $\mathcal{R}$. If $\mathcal{C}$ is separable, then $\mathcal{C}=\mathcal{C}'\otimes \mathscr{C}$, i.e. $\mathcal{C}=\langle(f(x),0),(0,g(x))  \rangle$ where $\mathcal{C}'=\langle f(x) \rangle$ with $f(x)$ is a right divisor of $x^{\alpha}-1$ and $ \mathscr{C}=\langle g(x) \rangle$ with $g(x)$ is a right divisor of $x^{\beta}-1$.
\begin{theorem}\label{sep}
Let $\mathcal{C}=\mathcal{C}'\otimes \mathscr{C}$  be an $\mathbb{F}_q\mathcal{R}$-skew cyclic code of length $(\alpha,\beta)$ where $\mathcal{C}'$ and $\mathscr{C}$ are skew cyclic codes over $\mathbb{F}_q$ and $\mathcal{R}$, respectively and $\mathcal{C}=\langle(f(x),0),(0,g(x)) \rangle$. Then $\mathcal{C}'^{\perp}\subseteq\mathcal{C}'$ and $\mathscr{C}^{\perp}\subseteq\mathscr{C}$ if and only if $\mathcal{C}^{\perp}\subseteq\mathcal{C}$.
\end{theorem}
\begin{proof}
Let $\mathcal{C}'^{\perp}\subseteq\mathcal{C}'$ and $\mathscr{C}^{\perp}\subseteq\mathscr{C}$. Let $c_1,c_1'\in \mathcal{C}'^{\perp}$ and $c_2,c_2'\in \mathscr{C}^{\perp}$. Then $c=(c_1,c_2), c'=(c_1',c_2')\in \mathcal{C}^{\perp}$.
 This implies $c_1\cdot c_1'=0$ in $\mathbb{F}_q$ and $ c_2\cdot c_2'=0$ in $\mathcal{R}$. Therefore,
$ c\cdot c'=uc_1c_1'+c_2c_2'=0$. Hence, $\mathcal{C}^{\perp}\subseteq\mathcal{C}$. \\
Converse follows directly from the definition.\qed
\end{proof}
\section{Gray Map}
Any arbitrary element of $\mathbb{F}_q\mathcal{R}$ can be uniquely written  as $(a,r)=(a,\xi_1r_1+\xi_2r_2)$, where $a\in\mathbb{F}_q$ and $r\in  \mathcal{R}$. Let $GL_2(\mathbb{F}_q)$ be the set of all $2\times 2$ invertible matrices over $\mathbb{F}_q$. Consider a Gray map (see \cite{Li20,Prakash} for similar mapping) $\phi:\mathcal{R}\longrightarrow \mathbb{F}_q^2$ defined by $$\phi(r_1+ur_2)=(r_1,r_2)M,$$ where $M\in GL_2(\mathbb{F}_q)$ such that $MM^T=\gamma I_2$, $M^T$ is the transpose matrix of $M$, $\gamma\in \mathbb{F}_q^*$ and $I_2$ is the identity matrix of order $2$. Then $\phi$ is a linear bijection and can be extended to $\mathcal{R}^n$ componentwise. Particularly, when $M=I_2$, we have the following Gray map
$\phi_1:\mathcal{R}\longrightarrow \mathbb{F}_q^2$
\begin{equation}\label{eq2}
	\phi_1(r_1+ur_2)=(r_1,r_2).
\end{equation}
Now, we define a Gray map $\varphi:\mathbb{F}_q\mathcal{R}\longrightarrow \mathbb{F}_q^3$ as follows:
$$\varphi(a,r)=(a,\phi(r))=(a,(r_1,r_2)M).$$ It is easy to verify that $\varphi$ is an $\mathbb{F}_q$-linear map and can be extended componentwise to $\mathbb{F}_q^{\alpha}\mathcal{R}^{\beta}$ in the following manner:
$$\varphi:\mathbb{F}_q^{\alpha}\mathcal{R}^{\beta}\longrightarrow \mathbb{F}_q^{\alpha+2\beta}$$
\begin{align*}
\varphi(a_0,a_1,\dots,a_{\alpha-1},r_0,r_1,\dots,r_{\beta-1})=(&a_0,a_1,\dots,a_{\alpha-1}, (r_{0,1},r_{0,2})M,(r_{1,1},r_{1,2})M,\\&\dots,(r_{\beta-1,1},r_{\beta-1,2})M),
\end{align*}
where $(a_0,a_1,\dots,a_{\alpha-1})\in \mathbb{F}_q^{\alpha}$ and $(r_0,r_1,\dots,r_{\beta-1})\in \mathcal{R}^{\beta}$ such that $r_j=\xi_1r_{j,1}+\xi_2r_{j,2}\in \mathcal{R}$ for $j=0,1,\dots,\beta-1$. The Lee weight of an element $(a,r)\in \mathbb{F}_q\mathcal{R}$ as $w_L(a,r)=w_H(a)+w_L(r)$ for any $(a,r)\in\mathbb{F}_q^{\alpha}\times\mathcal{R}^{\beta} $ where $w_H$ denotes the Hamming weight and $w_L$ denotes the Lee weight. Furthermore, the Lee distance between $c,c'\in \mathbb{F}_q^{\alpha}\times\mathcal{R}^{\beta} $ is defined as $d_L(c,c')=w_L(c-c')=w_H(\varphi(c-c'))$.
\begin{proposition}\label{prop1}
The Gray map $\varphi$ is an  $\mathbb{F}_q$-linear and distance preserving map from $\mathbb{F}_q^{\alpha}\mathcal{R}^{\beta}$ (Lee distance) to $\mathbb{F}_q^{\alpha+2\beta}$ (Hamming distance).
\end{proposition}
\begin{proof}
Let $c=(a,r), c'=(a',r')\in \mathbb{F}_q^{\alpha}\mathcal{R}^{\beta}$, where $a=(a_0,a_1,\dots,a_{\alpha-1})\in\mathbb{F}_q^{\alpha}$, $r=(r_0,r_1,\dots,r_{\beta-1})\in\mathcal{R}^{\beta}$, $a'=(a_0',a_1',\dots,a_{\alpha-1}')\in\mathbb{F}_q^{\alpha}$, $r'=(r_0',r_1',\dots,r_{\beta-1}')\in\mathcal{R}^{\beta}$ such that
$r_j=\xi_1r_{j,1}+\xi_2r_{j,2}$ and $r_j'=\xi_1r_{j,1}'+\xi_2r_{j,2}'$ for $j=0,1,\dots,\beta-1$. Then
\begin{align*}
	\varphi(c+c')=&\varphi(a+a',r+r')\\
	=&\varphi(a_0+a_0',a_1+a_1',\dots, a_{\alpha-1}+a_{\alpha-1}',r_0+r_0',r_1+r_1',\dots,r_{\beta-1}+r_{\beta-1}')\\
	=&\varphi(a_0+a_0',a_1+a_1',\dots, a_{\alpha-1}+a_{\alpha-1}',\xi_1(r_{0,1}+r_{0,1}')+\xi_2(r_{0,2}+r_{0,2}'),\\&\xi_1(r_{1,1}+r_{1,1}')+\xi_2(r_{1,2}+r_{1,2}'),\dots,\xi_1(r_{\beta-1,1}+r_{\beta-1,1}')+\xi_2(r_{\beta-1,2}+r_{\beta-1,2}')\\
	=&(a_0+a_0',a_1+a_1',\dots, a_{\alpha-1}+a_{\alpha-1}',(r_{0,1}+r_{0,1}',r_{0,2}+r_{0,2}')M,(r_{1,1}+r_{1,1}',r_{1,2}\\&+r_{1,2}')M,\dots,(r_{\beta-1,1}+r_{\beta-1,1}',r_{\beta-1,2}+r_{\beta-1,2}')M)\\
	=&(a_0,a_1,\dots,a_{\alpha-1}, (r_{0,1},r_{0,2})M,(r_{1,1},r_{1,2})M,\dots,(r_{\beta-1,1},r_{\beta-1,2})M)+(a_0',a_1',\dots,\\&a_{\alpha-1}', (r_{0,1}',r_{0,2}')M,(r_{1,1}',r_{1,2}')M,\dots,(r_{\beta-1,1}',r_{\beta-1,2}')M)\\
	=&\varphi(a,r)+\varphi(a',r')\\
	=&\varphi(c)+\varphi(c').
\end{align*} Now, for any $\lambda\in \mathbb{F}_q$, we have
\begin{align*}
	\varphi(\lambda c)&=\varphi(\lambda a,\lambda r)\\&=(\lambda(a_0,a_1,\dots,a_{\alpha-1}),\lambda (r_{0,1},r_{0,2})M,\lambda(r_{1,1},r_{1,2})M,\dots,\lambda(r_{\beta-1,1},r_{\beta-1,2})M)\\
	&=\lambda(a_0,a_1,\dots,a_{\alpha-1}, (r_{0,1},r_{0,2})M,(r_{1,1},r_{1,2})M,\dots,(r_{\beta-1,1},r_{\beta-1,2})M)\\
	&=\lambda\varphi(a,r)\\
	&=\lambda\varphi(c).
\end{align*} Thus, $\varphi$ is an $\mathbb{F}_q$-linear map. Moreover, $d_L(c,c')=w_L(c-c')=w_H(\varphi(c-c'))=w_H(
\varphi(c)-\varphi(c'))=d_H(\varphi(c),\varphi(c'))$. Therefore, $\varphi$ is a distance preserving map.\qed
\end{proof}
\begin{theorem}
If $\mathcal{C}$ is an $[n,k,d_L]$ $\mathbb{F}_q\mathcal{R}$-linear code, then $\varphi(\mathcal{C})$ is a $[\alpha+2\beta,k,d_H]$ linear code over $\mathbb{F}_q$.
\end{theorem}
\begin{proof}
It follows directly from Proposition \ref{prop1} and the definition of the Gray map.
\end{proof}
Now, the following result shows the Gray map $\varphi$ preserves the orthogonality.\qed
\begin{theorem}
Let $\mathcal{C}$ be an $\mathbb{F}_q\mathcal{R}$-skew cyclic code of length $(\alpha,\beta)$. Then $\varphi(\mathcal{C})^\perp=\varphi(\mathcal{C}^\perp)$. Further, $\mathcal{C}$ is self-dual if and only if $\varphi(\mathcal{C})$ is self-dual.
\end{theorem}
\begin{proof}
Let $\mathcal{C}$ be an $\mathbb{F}_q\mathcal{R}$-skew cyclic code of length $(\alpha,\beta)$. Assume that $l=(x_0,x_1,\dots,x_{\alpha-1},\\y_0,y_1,\dots,y_{\beta-1})\in \mathcal{C}^\perp$ where $y_j=\xi_1a_j+\xi_2b_j$ for $0\leq j\leq \beta-1$. Then $\varphi(l)\in\varphi(\mathcal{C^\perp})$. To show that $\varphi(l)\in \varphi(\mathcal{C})^\perp$, we consider $l'=(x_0',x_1',\dots,x_{\alpha-1}',y_0',y_1',\dots,y_{\beta-1}')\in\mathcal{C}$ where $y_j'=\xi_1a_j'+\xi_2b_j'$ for $0\leq j\leq \beta-1$. Now, $l\cdot l'=0$ implies that
\begin{align*}
l\cdot l'&=u\sum_{i=0}^{\alpha -1}x_{i}x_{i}'+\sum_{j=0}^{\beta -1}y_{j}y_{j}'\\
&=u\sum_{i=0}^{\alpha -1}x_{i}x_{i}'+\sum_{j=0}^{\beta -1}(\xi_1a_ja_j'+\xi_2b_jb_j')=0.
\end{align*}
Further, it gives $\sum_{i=0}^{\alpha -1}x_{i}x_{i}'=0$ and $\sum_{j=0}^{\beta -1}(a_ja_j'+b_jb_j')=0$. Next, we have
\begin{align*}
   \varphi(l)&= (x_0,x_1,\dots,x_{\alpha-1},(a_0,b_0)M,(a_1,b_1)M,\cdots,(a_{\beta-1},b_{\beta-1})M)\\
   &=(x_0,x_1,\dots,x_{\alpha-1},c_0M,c_1M,\cdots,c_{\beta-1}M)\\
   \varphi(l')&= (x_0',x_1',\dots,x_{\alpha-1}',(a_0',b_0')M,(a_1',b_1')M,\cdots,(a_{\beta-1}',b_{\beta-1})'M)\\
   &=(x_0',x_1',\dots,x_{\alpha-1}',c_0'M,c_1'M,\cdots,c_{\beta-1}'M)
\end{align*} where $c_j=(a_j,b_j)$ and $c_j'=(a_j',b_j')$ for $0\leq j\leq \beta-1$. Also,
\begin{align*}
\varphi(l)\cdot \varphi(l')&=\varphi(l) \varphi(l')^T=u\sum_{i=0}^{\alpha -1}x_{i}x_{i}'+\sum_{j=0}^{\beta -1}y_{j}MM^Ty_{j}'^T\\
&=u\sum_{i=0}^{\alpha -1}x_{i}x_{i}'+\gamma\sum_{j=0}^{\beta -1}y_{j}y_{j}'^T\\
&=u\sum_{i=0}^{\alpha -1}x_{i}x_{i}'+\gamma\sum_{j=0}^{\beta -1}y_{j}y_{j}'^T\\
&=u\sum_{i=0}^{\alpha -1}x_{i}x_{i}'+\gamma\sum_{j=0}^{\beta -1}(a_ja_j'+b_jb_j').
\end{align*}
Then, $\varphi(l)\in \varphi(\mathcal{C})^\perp$ and hence $\varphi(\mathcal{C}^\perp)\subseteq \varphi(\mathcal{C})^\perp$. As $\varphi$ is a bijection, $|\varphi(\mathcal{C}^\perp)|=|\varphi(\mathcal{C})^\perp|$. Thus, $\varphi(\mathcal{C}^\perp)= \varphi(\mathcal{C})^\perp$.\\
Let $\mathcal{C}$ be self-dual, i.e., $\mathcal{C}=\mathcal{C}^\perp$. Then $\varphi(\mathcal{C})=\varphi(\mathcal{C}^\perp)= \varphi(\mathcal{C})^\perp$, i.e., $ \varphi(\mathcal{C})$ is self-dual. Conversely, let $ \varphi(\mathcal{C})$ be self-dual.  Then $\varphi(\mathcal{C})=\varphi(\mathcal{C}^\perp)= \varphi(\mathcal{C})^\perp$. Since, $\varphi$ is a bijection, $\mathcal{C}=\mathcal{C}^\perp$. Hence, $\mathcal{C}$ is self-dual.\qed
\end{proof}
\begin{theorem}
Let $\mathcal{C}$ be an $\mathbb{F}_q\mathcal{R}$-skew cyclic code of length $(\alpha,\beta)$. Then $\varphi(\mathcal{C})=\mathcal{C}'\otimes \mathcal{C}_1\otimes\mathcal{C}_2$, $\mathcal{C}'$ is a skew cyclic code of length $\alpha$ in $\mathbb{F}_q[x;\Theta]/\langle x^{\alpha}-1\rangle$ and $\mathcal{C}_1$, $\mathcal{C}_2$ are skew cyclic codes of length $\beta$ over $\mathbb{F}_q$. Furthermore, $|\varphi(\mathcal{C})|=|\mathcal{C}'||\mathcal{C}_1||\mathcal{C}_2|$.
\end{theorem}
\begin{proof}
Consider an element $c=(a,r)\in \mathcal{C}$ such that $a=(a_0,a_1,\dots,a_{\alpha-1})\in\mathbb{F}_q^{\alpha}$ and $r=(r_0,r_1,\dots,r_{\beta-1})\in \mathcal{R}^{\beta}$ where $r_i=b_i+uc_i$ for $i=0,1,\dots,\beta-1$. Now, we define
$$\mathcal{C}':=(a_0,a_1,\dots,a_{\alpha-1}),$$
$$\mathcal{C}_1:=(b_0,b_1,\dots,b_{\beta-1}),$$
and $$\mathcal{C}_2:=(c_0,c_1,\dots,c_{\beta-1}).$$
Next, a given codeword $c'=(a_0,a_1,\dots,a_{\alpha-1})\in\mathcal{C}'$ corresponds to a codeword $c=(a_0,a_1,\dots,a_{\alpha-1},b_0+uc_0,b_1+uc_1,\dots,b_{\beta-1}+uc_{\beta-1})\in\mathcal{C}$. As $\mathcal{C}$ is an $\mathbb{F}_q\mathcal{R}$-skew cyclic code, the $\sigma$-skew cyclic shift of $c$ is given by $\sigma(c)=(\Theta(a_{\alpha-1}),\Theta(a_0),\dots,\Theta(a_{\alpha-2}),\theta(b_{\beta-1}+uc_{\beta-1}),\theta(b_0+uc_0),\theta(b_1+uc_1),\dots,\theta(b_{\beta-2}+uc_{\beta-2}))\in \mathcal{C}$. Therefore, $(\Theta(a_{\alpha-1}),\Theta(a_0),\dots,\\\Theta(a_{\alpha-2}))\in \mathcal{C}'$. Hence, $\mathcal{C}'$ is a skew cyclic code of length $\alpha$ in $\mathbb{F}_q[x;\Theta]/\langle x^{\alpha}-1\rangle.$ Similarly, we get $\mathcal{C}_1$, $\mathcal{C}_2$ are skew cyclic codes of length $\beta$ over $\mathbb{F}_q$. Thus, $\varphi(\mathcal{C})=\mathcal{C}'\otimes \mathcal{C}_1\otimes\mathcal{C}_2$. As $\varphi$ is bijective, $|\mathcal{C}|=|\varphi(\mathcal{C})|$ and hence $|\mathcal{C}|=|\mathcal{C}'||\mathcal{C}_1||\mathcal{C}_2|$.\qed
\end{proof}
Now, we present an example of an $\mathbb{F}_q\mathcal{R}$-skew cyclic code. To compute Gray image $\varphi(\mathcal{C})$ of an $\mathbb{F}_q\mathcal{R}$-skew cyclic code $\mathcal{C}$, first we find the corresponding generator polynomials $f(x)$ of skew cyclic code $\mathcal{C}'$ of length $\alpha$ over $\mathbb{F}_q$ and $g(x)$ of skew cyclic code $\mathscr{C}$ of length $\beta$ over $\mathcal{R}$ with $g(x)=g_1(x)\xi_1+g_2(x)\xi_2$, where $g_1(x)$ and $g_2(x)$ are the generator polynomials of skew cyclic codes $\mathcal{C}_1$  and $\mathcal{C}_2$ over $\mathbb{F}_q$ respectively. Then, with the help of Magma computational algebra system \cite{Bosma}, we find $\varphi(\mathcal{C})$.
\begin{example}
	Let $q=9$, $\alpha=18$, $\beta=36$ and $\mathcal{R}=\mathbb{F}_9[u]/\langle u^2-u\rangle$ where $\mathbb{F}_9=\mathbb{F}_3(w)$ and $w^2+2w+2=0$. In $\mathbb{F}_9$, the Frobenius automorphism $\Theta:\mathbb{F}_9\longrightarrow \mathbb{F}_9$ is defined by $\Theta(a)=a^3$ for all $a\in \mathbb{F}_9$. Therefore, $\mathbb{F}_9[x;\Theta]$ is a skew polynomial ring. In $\mathbb{F}_9[x;\Theta]$, we have
	\begin{align*}
	    x^{18}-1&=x^{15} + w^5x^{14} + x^{13} + 2x^{12} + x^9 + w^5x^8 + x^7 + 2x^6 + x^3 +
        w^5x^2 + x + 2)(x^3\\&+w^3x^2+x+1).\\
	    x^{36}-1&=(x^{35} + w^2x^{34} + x^{33} + w^2x^{32} + x^{31} + w^2x^{30} + x^{29} + w^2x^{28} +
        x^{27} + w^2x^{26} + x^{25}\\& + w^2x^{24} + x^{23} + w^2x^{22} + x^{21} + w^2x^{20} +
        x^{19} + w^2x^{18} + x^{17} + w^2x^{16} + x^{15} + w^2x^{14} \\&+ x^{13} + w^2x^{12} +
        x^{11} + w^2x^{10} + x^9 + w^2x^8 + x^7 + w^2x^6 + x^5 + w^2x^4 + x^3 +
        w^2x^2\\& + x + w^2)(x+w^2).\\
	            &=w^6x^{34} + 2x^{33} + w^5x^{32} + 2x^{31} + x^{30} + w^2x^{28} + x^{27} + wx^{26}
        + x^{25} + 2x^{24} + w^6x^{22} \\&+ 2x^{21}+ w^5x^{20} + 2x^{19} + x^{18} +
        w^2x^{16} + x^{15} + wx^{14} + x^{13} + 2x^{12} + w^6x^{10} + 2x^9 \\&+ w^5x^8 +
        2x^7 + x^6 + w^2x^4 + x^3 + wx^2 + x + 2)(w^2x^2+x+1).
	\end{align*}
	Now, let $f(x)=x^3+w^3x^2+x+1,$ $g_1(x
	)=x+w^2$ and $g_2(x)=w^2x^2+x+1$. Then $\mathcal{C}$ is an $\mathbb{F}_q\mathcal{R}$-skew cyclic code of length $(18,36)$.
Let 	
	\begin{equation*}\label{eq3}
	M=
	\begin{pmatrix}
	1&1\\
	1&-1
	\end{pmatrix}\in GL_2(\mathbb{F}_9)
	\end{equation*}
	such that $MM^T=2I_2$. Then $\varphi(\mathcal{C})$ is a $[90,84,4]_9$ linear code over $\mathbb{F}_9$ which is BKLC (best-known linear code) as per the database \cite{Grassl}.
	\end{example}
\section{Quantum Codes from $\mathbb{F}_q\mathcal{R}$-Skew Cyclic Codes}
Recall that the set of $n$-fold tensor product $(\mathbb{C}^q)^{\otimes n}=\underbrace{\mathbb{C}^q\otimes \mathbb{C}^q\otimes\cdots  \otimes \mathbb{C}^q}_{n\rm\ times}$ is a $q^n$-dimensional Hilbert space, where $\mathbb{C}^q$ is a $q$-dimensional complex Hilbert space. A $q$-ary quantum code of length $n$, dimension $k$ and minimum distance $d$ is denoted by $[[n,k,d]]_q$, is a $q^k$-dimensional subspace of $(\mathbb{C}^q)^{\otimes n}$. A quantum code $[[n,k,d]]_q$ is known as quantum MDS (maximum-distance-separable) if it attains the singleton bound $k+2d\leq n+2$. Note that a quantum code $[[n,k,d]]_q$ is said to be better than $[[n',k',d']]_q$ if any one of the following or both hold:
\begin{enumerate}
    \item $d>d'$ when the code rate $\frac{k}{n}=\frac{k'}{n'}$ (Larger distance with same code rate).
    \item $\frac{k}{n}>\frac{k'}{n'}$ when the distance $d=d'$ (Larger code rate with same distance).
\end{enumerate}
The relationship between quantum information and classical information is a matter currently receiving much attention from researchers. One of the fascinating developments in the study of linear codes is the construction of quantum codes from classical codes. The first quantum code was independently studied by Shor \cite{Shor} and Steane \cite{Steane}. However, the construction of quantum codes from classical codes, their existence proofs and correction methods were discussed by Calderbank et al. \cite{Calderbank}.
One of the widely used techniques to obtain quantum codes from classical linear codes is the CSS construction (Lemma \ref{lemma css}) in which dual containing linear codes play an instrumental role.
\begin{lemma}\cite[Theorem 3]{Grassl04} \label{lemma css}
Let $\mathcal{C}_{1}=[n,k_{1},d_{1}]_{q}$ and $\mathcal{C}_{2}=[n,k_{2},d_{2}]_{q}$ be two linear codes over $GF(q)$ with $\mathcal{C}_{2}^{\perp}\subseteq \mathcal{C}_{1}$. Then there exists a quantum error-correcting code $\mathcal{C}=[[n,k_{1}+k_{2}-n,d]]$ where $d=\min\{w(v): v\in (\mathcal{C}_{1}\backslash \mathcal{C}_{2}^{\perp})\cup (\mathcal{C}_{2}\backslash\mathcal{C}_{1}^{\perp})\}\geq \min\{ d_{1},d_{2}\}$. Further, if $\mathcal{C}_{1}^{\perp}\subseteq \mathcal{C}_{1},$ then there exists a quantum error-correcting code $\mathcal{C}=[[n,2k_{1}-n,d_{1}]]$, where $d_{1}=\min\{w(v):v\in \mathcal{C}_{1} \backslash \mathcal{C}_{1}^{\perp}\}.$
\end{lemma}
Now, we recall from \cite{Boucher2} that a skew cyclic code of length $n$ over $\mathbb{F}_q$ is a principally generated left $\mathbb{F}_q[x;\Theta]$-submodule of $\frac{\mathbb{F}_q[x;\Theta]}{\langle x^n-1\rangle}$ and $\mathcal{C}=\langle g(x)\rangle$ where $x^n-1=h(x)g(x)$, i.e., $g(x)$ is a right divisor of $x^n-1$. Its dual $\mathcal{C}^{\perp}$ is also a skew cyclic code of length $n$ and $\mathcal{C}^{\perp}=\langle h^{\dagger}(x)\rangle$ where $h^{\dagger}(x)=h_{n-r}+\Theta(h_{n-r-1})x+\dots+\Theta^{n-r}(h_0)x^{n-r}$ for $h(x)=h_0+h_1x+\dots+h_{n-r}x^{n-r}$. If $\Theta=id$, then $h^{\dagger}(x)=h^*(x)$ where $h^*(x)=h_{n-r}+h_{n-r-1}x+\dots+h_{0}x^{n-r}$. Further, we recall the dual containing property for skew cyclic codes.
\begin{lemma}\label{thm constadual}\cite{Islam18b}
Let $\gcd(\alpha, |\langle \Theta\rangle|)=1$ and $\mathcal{C}=\langle f(x)\rangle$ be a skew cyclic
code of length $\alpha$ over $\mathbb{F}_q$.
Then $\mathcal{C}^{\perp}\subseteq \mathcal{C}$ if and only if $x^{\alpha}-1\equiv 0 \pmod{f(x)f^*(x)},$ where $f^*(x)$ is the reciprocal polynomial of $f(x)$.
\end{lemma}

\begin{lemma}\label{lemm dualcyc}\cite{bag,Li20}
Let $|\langle \Theta\rangle|$ divide  $\alpha$ and $\mathcal{C}=\langle g(x)\rangle$ be a skew cyclic code of length $\alpha$ over $\mathbb{F}_q$ where $x^\alpha-1=h(x)g(x)$. Then $\mathcal{C}^{\perp}\subseteq \mathcal{C}$ if and only if $h^{\dagger}(x)h(x)$ is divisible by $x^\alpha-1$ from the right side.
\end{lemma}
Now, we derive the necessary and sufficient conditions for  $\mathbb{F}_q\mathcal{R}$-skew cyclic codes to contain their duals.
\begin{theorem} \label{dualcon}
Let $|\langle\Theta\rangle|$ divide $\alpha$, $|\langle\theta\rangle|$ divide $\beta$ and $\mathcal{C}=\mathcal{C}'\otimes\mathscr{C}$ be an $\mathbb{F}_q\mathcal{R}$-skew cyclic code of length $(\alpha,\beta)$ where $\mathcal{C}'=\langle f(x)\rangle$ and $\mathscr{C}=\xi_1\mathcal{C}_1\oplus\xi_2\mathcal{C}_2=\langle g(x)\rangle$  are  skew cyclic codes over $\mathbb{F}_q$ and $\mathcal{R}$ respectively, with $g(x)=
 \xi_1g_1(x)+\xi_2g_2(x)$ and $x^{\alpha}-1=h(x)f(x)$, $x^{\beta}-1=h_i(x)g_i(x)$.
Then $\mathcal{C}^{\perp}\subseteq \mathcal{C}$ if and only if $h^{\dagger}(x)h(x)$ and $h^{\dagger}_i(x)h_i(x)$ are divisible by $x^\alpha-1$ and $x^\beta-1$, respectively for $i=1,2$ from the right side.
\end{theorem}
\begin{proof}
Let $\mathcal{C}^{\perp}\subseteq \mathcal{C}$. Then by Theorem \ref{sep}, $\mathcal{C}'^{\perp}\subseteq\mathcal{C}'$ and $\mathscr{C}^{\perp}\subseteq\mathscr{C}$. Also, $\mathcal{C}_1^{\perp}\subseteq\mathcal{C}_1$ and $\mathcal{C}_2^{\perp}\subseteq\mathcal{C}_2$ where $\mathcal{C}_1$ and $\mathcal{C}_2$ are skew cyclic codes of length $\beta$ over $\mathbb{F}_q$. Then from Lemma \ref{lemm dualcyc}, we have  $h^{\dagger}(x)h(x)$ and $h^{\dagger}_i(x)h_i(x)$ are divisible by $x^\alpha-1$ and $x^\beta-1$, respectively from the right side for $i=1,2$.\\ Conversely, suppose $h^{\dagger}(x)h(x)$ and $h^{\dagger}_i(x)h_i(x)$ are divisible by $x^\alpha-1$ and $x^\beta-1$, respectively from the right side for $i=1,2$. Then by Lemma \ref{lemm dualcyc}, we have $\mathcal{C}'^{\perp}\subseteq\mathcal{C}'$, $\mathcal{C}_1^{\perp}\subseteq\mathcal{C}_1$ and $\mathcal{C}_2^{\perp}\subseteq\mathcal{C}_2$. This implies $\mathscr{C}^{\perp}\subseteq\mathscr{C}$. Therefore, by Theorem \ref{sep},  $\mathcal{C}^{\perp}\subseteq \mathcal{C}$. \qed
\end{proof}
Finally, we present our main result for the construction of quantum codes from $\mathbb{F}_q\mathcal{R}$-skew cyclic codes with the help of Theorem \ref{dualcon}.
\begin{theorem}\label{thquantum}
Let $|\langle\Theta\rangle|$ divide $\alpha$, $|\langle\theta\rangle|$ divide $\beta$  and $\mathcal{C}=\mathcal{C}'\otimes\mathscr{C}$ be an $\mathbb{F}_q\mathcal{R}$-skew cyclic code of length $(\alpha,\beta)$ with $x^{\alpha}-1=h(x)f(x)$, $x^{\beta}-1=h_i(x)g_i(x)$. If $h^{\dagger}(x)h(x)$ and $h^{\dagger}_i(x)h_i(x)$ are divisible by $x^\alpha-1$ and $x^\beta-1$, respectively from the right side for $i=1,2$, then there exists a quantum code $[[\alpha+2\beta,2k-(\alpha+2\beta),d_H]]_q$ where $d_H$ denotes the Hamming distance and $k$ denotes the dimension of the code $\varphi(\mathcal{C})$.
\end{theorem}
\begin{proof}
Let $h^{\dagger}(x)h(x)$ and $h^{\dagger}_i(x)h_i(x)$ be divisible by $x^\alpha-1$ and $x^\beta-1$, respectively from the right side for $i=1,2$. Then by Theorem \ref{dualcon}, we have $\mathcal{C}^{\perp}\subseteq \mathcal{C}$. As $\varphi(\mathcal{C})=\mathcal{C}'\otimes \mathcal{C}_1\otimes\mathcal{C}_2$, it is easy to check that $\varphi(\mathcal{C})^\perp=\mathcal{C}'^\perp\otimes \mathcal{C}_1^\perp\otimes\mathcal{C}_2^\perp$. Also, $\mathcal{C}'^{\perp}\subseteq\mathcal{C}'$, $\mathcal{C}_1^{\perp}\subseteq\mathcal{C
}_1$ and $\mathcal{C}_2^{\perp}\subseteq\mathcal{C}_2$. Therefore, $\varphi(\mathcal{C})$ is a dual containing linear code with parameters $[\alpha+2\beta,k,d_H]$. Hence, by Lemma \ref{lemma css}, there exists a quantum code with parameters $[[\alpha+2\beta,2k-(\alpha+2\beta),d_H]]_q$. \qed
\end{proof}
Now, with the help of Corollary \ref{cor1}, Theorem \ref{thquantum} can also be stated in the following form:
\begin{corollary}\label{coro-quantum}
Let $|\langle\theta\rangle|$ divide $\beta$, $\gcd(\alpha, |\langle \Theta\rangle|)=1$  and $\mathcal{C}=\mathcal{C}'\otimes\mathscr{C}$ be an $\mathbb{F}_q\mathcal{R}$-skew cyclic code of length $(\alpha,\beta)$ where $\mathcal{C}'$ is a cyclic code over $\mathbb{F}_q$ and $\mathscr{C}$  is a  skew cyclic code over $\mathcal{R}$ with $x^{\alpha}-1=h(x)f(x)$, $x^{\beta}-1=h_i(x)g_i(x)$. If $x^\alpha-1$ is divisible by  $f(x)f^*(x)$ and and $h^{\dagger}_i(x)h_i(x)$ is divisible by $x^\beta-1$ from the right side for $i=1,2$, then there exists a quantum code $[[\alpha+2\beta,2k-(\alpha+2\beta),d_H]]_q$.
\end{corollary}
\begin{remark}
Notice that the length of the quantum code obtained from the codes over the ring $\mathcal{R}$
must be an integral multiple of $2$, whereas the code length in Theorem \ref{thquantum} has no such
restriction, i.e., we can find code of any length $\alpha+2\beta$ with some suitable choices of
 $\alpha$ and $\beta$. For example, to construct a code of length $60$, there are finitely many choices for $\alpha$ and $\beta$ such that $\alpha+2\beta=60$. However, in the case of $\mathcal{R}$, the only choice is $\beta=30$. This
is one of the advantages of studying the $\mathbb{F}_q\mathcal{R}$-skew cyclic codes in quantum
code constructions.
\end{remark}
Now, with the help of our derived results, we construct several new and better quantum codes than the existing ones appearing in \cite{bag,Edel,Verma21}. All our computations in the examples are carried out using the Magma
computation system \cite{Bosma}.
\begin{example}
Let $q=9$, $\alpha=49$, $\beta=36$ and $\mathcal{R}=\mathbb{F}_9[u]/\langle u^2-u\rangle$ where $\mathbb{F}_9=\mathbb{F}_3(w)$ and $w^2+2w+2=0$.  In $\mathbb{F}_9$, the Frobenius automorphism $\Theta:\mathbb{F}_9\longrightarrow \mathbb{F}_9$ is defined by $\Theta(a)=a^3$ for all $a\in \mathbb{F}_9$. Therefore, $\mathbb{F}_9[x;\Theta]$ is a skew polynomial ring. In $\mathbb{F}_9[x;\Theta]$, we have
	\begin{align*}
	    x^{36}-1&= w^2x^{34} + 2x^{33} + w^7x^{32} + 2x^{31} + x^{30} + w^6x^{28} + x^{27} +
        w^3x^{26} + x^{25} + 2x^{24} + w^2x^{22}\\& + 2x^{21} + w^7x^{20} + 2x^{19} + x^{18}
        + w^6x^{16} + x^{15} + w^3x^{14} + x^{13} + 2x^{12} + w^2x^{10} + 2x^9\\& +
        w^7x^8 + 2x^7 + x^6 + w^6x^4 + x^3 + w^3x^2 + x + 2)
        (w^6x^2 + x + 1),\\
	    &=(w^6x^{34} + 2x^{33} + w^5x^{32} + 2x^{31} + x^{30} + w^2x^{28} + x^{27} + wx^{26}
        + x^{25} + 2x^{24} + w^6x^{22} \\&+ 2x^{21} + w^5x^{20} + 2x^{19} + x^{18} +
        w^2x^{16} + x^{15} + wx^{14} + x^{13} + 2x^{12} + w^6x^{10} + 2x^9\\& + w^5x^8 +
        2x^7 + x^6 + w^2x^4 + x^3 + wx^2 + x + 2)(
        w^2x^2 + x + 1).
\end{align*} In $\mathbb{F}_9[x]$, we have	
\begin{align*}
x^{49}-1&=(x + 2)(x^3 + wx^2 + w^7x + 2)(x^3 + w^3x^2 + w^5x + 2)(x^{21} + wx^{14} + w^7x^7 + 2)\\&(x^{21} + w^3x^{14} + w^5x^7 + 2).
\end{align*}
	Now, let $f(x)=x^3 + w^3x^2 + w^5x + 2,$ $g_1(x
	)=w^6x^2 + x + 1$ and $g_2(x)=w^2x^2+x+1$. Then $\mathcal{C}$ is an $\mathbb{F}_q\mathcal{R}$-skew cyclic code of length $(49,36)$.
Let 	
	\begin{equation*}\label{eq3}
	M=
	\begin{pmatrix}
	1&1\\
	1&-1
	\end{pmatrix}\in GL_2(\mathbb{F}_9)
	\end{equation*}
	such that $MM^T=2I_2$. Then $\varphi(\mathcal{C})$ is a $[121,114,4]_9$ linear code over $\mathbb{F}_9$. Now, we have
\begin{align*}
h_1(x)&=w^2x^{34} + 2x^{33} + w^7x^{32} + 2x^{31} + x^{30} + w^6x^{28} + x^{27} +
        w^3x^{26} + x^{25} + 2x^{24} + w^2x^{22}\\& + 2x^{21} + w^7x^{20} + 2x^{19} + x^{18}
        + w^6x^{16} + x^{15} + w^3x^{14} + x^{13} + 2x^{12} + w^2x^{10} + 2x^9\\& +
        w^7x^8 + 2x^7 + x^6 + w^6x^4 + x^3 + w^3x^2 + x + 2, \\
h_1^{\dagger}(x)&=2x^{34} + x^{33} + w^3x^{32} + x^{31} + w^6x^{30} + x^{28} + 2x^{27} + w^7x^{26} +
        2x^{25} + w^2x^{24} + 2x^{22}\\& + x^{21} + w^3x^{20} + x^{19} + w^6x^{18} + x^{16} +
        2x^{15} + w^7x^{14} + 2x^{13} + w^2x^{12} + 2x^{10} + x^9 \\&+ w^3x^8 + x^7 +
        w^6x^6 + x^4 + 2x^3 + w^7x^2 + 2x + w^2,\\
h_1^{\dagger}(x)h_1(x)&=(w^6x^{32} + w^5x^{31} + wx^{30} + w^6x^{29} + w^7x^{28} + w^2x^{27} + 2x^{26} + 2x^{25}
+ 2x^{24} \\&+ w^6x^{23} + w^7x^{22} + w^2x^{21} + wx^{20} + w^7x^{19} + w^6x^{18} +
w^2x^{14} + wx^{13} + w^5x^{12}\\& + w^2x^{11} + w^3x^{10} + w^6x^9 + x^8 + x^7 + x^6 +
w^2x^5 + w^3x^4 + w^6x^3 + w^5x^2 + w^3x \\&+ w^2)(x^{36}-1),
\end{align*}
and
\begin{align*}
 h_2(x)&=w^6x^{34} + 2x^{33} + w^5x^{32} + 2x^{31} + x^{30} + w^2x^{28} + x^{27} + wx^{26}
        + x^{25} + 2x^{24} + w^6x^{22} \\&+ 2x^{21} + w^5x^{20} + 2x^{19} + x^{18} +
        w^2x^{16} + x^{15} + wx^{14} + x^{13} + 2x^{12} + w^6x^{10} + 2x^9\\& + w^5x^8 +
        2x^7 + x^6 + w^2x^4 + x^3 + wx^2 + x + 2, \\ h_2^{\dagger}(x)&=2x^{34} + x^{33} + wx^{32} + x^{31} + w^2x^{30} + x^{28} + 2x^{27} + w^5x^{26} +
        2x^{25} + w^6x^{24} + 2x^{22} \\&+ x^{21} + wx^{20} + x^{19} + w^2x^{18} + x^{16} +
        2x^{15} + w^5x^{14} + 2x^{13} + w^6x^{12} + 2x^{10} + x^9 \\&+ wx^8 + x^7 +
        w^2x^6 + x^4 + 2x^3 + w^5x^2 + 2x + w^6,\\
 h_2^{\dagger}(x)h_2(x)&=(w^2x^{32} + w^7x^{31} + w^3x^{30} + w^2x^{29} + w^5x^{28} + w^6x^{27} + 2x^{26} +
2x^{25} + 2x^{24} \\&+ w^2x^{23} + w^5x^{22} + w^6x^{21} + w^3x^{20} + w^5x^{19} +
w^2x^{18} + w^6x^{14} + w^3x^{13} + w^7x^{12} \\&+ w^6x^{11} + wx^{10} + w^2x^9 + x^8 +
x^7 + x^6 + w^6x^5 + wx^4 + w^2x^3 + w^7x^2 + wx \\&+ w^6)(x^{36}-1).
\end{align*}
Therefore, $h_1^{\dagger}(x)h_1(x)$ and  $h_2^{\dagger}(x)h_2(x)$ are divisible by $x^{36}-1$ from the right sides. Also, $x^{49}-1$ is divisible by $f(x)f^*(x)$. Hence, by Theorem \ref{coro-quantum}, we have a quantum code with parameters $[[121,107,4]]_{9}$ which has same
length and minimum distance but larger code rate than existing code $[[121,106,4]]_9$ given by \cite{Edel}.
	\end{example}


\begin{table}[ht]
\renewcommand{\arraystretch}{1.8}
\begin{center}
\caption{New quantum codes from $\mathbb{F}_q\mathcal{R}$-skew cyclic codes}

\begin{tabular}{|c|c|c|c|c|c|c|c|c|c|c|c|c|c|c|}

\hline
$q$ & $(\alpha,\beta)$ &  $[f(x),g_1(x),g_{2}(x)]$ & $\varphi(\mathcal{C})$ & Obtained  & Existing  \\

  &   &    &   & Codes  & Codes \\

\hline

$9$ & $(49,40)$ & $(1ww^72, w^61, 1w^2)$   & $[129,124,3]_9$ & $[[129,119,3]]_9$ & $[[129, 118,3]]_9$ \cite{Edel} \\

$9$ & $(49,36)$ & $(1w^3w^52, w^611, w^211)$   & $[121,114,4]_9$ & $[[121,107,4]]_9$ & $[[121,106,4]]_9$ \cite{Edel} \\

$9$ & $(75,30)$ & $(1w^31, w^7w^30, w^61)$   & $[285,280,3]_9$ & $[[285,275,3]]_9$ & $[[286,275,3]]_9$ \cite{Edel} \\

$9$ & $(729,20)$ & $(12021, w1, w^5w^31)$   & $[769,762,3]_9$ & $[[769,755,3]]_9$ & $[[769,745,3]]_9$ \cite{Edel} \\

$25$ & $(21,30)$ & $(1ww^{17}4, w^{15}0w^3,w^{8}1)$   & $[81,75,3]_{25}$ & $[[81,69,3]]_{25}$ & $[[80,56,3]]_{25}$ \cite{bag} \\

$25$ & $(55,32)$ & $(131, w^{9}1,1w^{21}1)$   & $[119,114,3]_{25}$ & $[[119,109,3]]_{25}$ & $[[120,106,3]]_{25}$ \cite{Verma21} \\

$49$ & $(75,36)$ & $(1ww^{9}6, w^{21}51, w^{39}41)$   & $[147,140,3]_{49}$ & $[[147,133,3]]_{49}$ & $[[144,126,3]]_{49}$ \cite{Verma21} \\

\hline
\end{tabular}\label{tab2}
\end{center}
\end{table}
In Table $1$, we present some quantum codes from $\mathbb{F}_q\mathcal{R}$-skew cyclic code for $q=9,25,49$. To compute their Gray images for $q=9,25,49$, we use
\begin{equation*}\label{eq3}
	M=
	\begin{pmatrix}
	1&1\\
	1&-1
	\end{pmatrix}\in GL_2(\mathbb{F}_q)
	\end{equation*} satisfying $MM^T=2I_2$.
Let  $\mathcal{C}=\mathcal{C}'\otimes\mathscr{C}$ be an $\mathbb{F}_q\mathcal{R}$-skew cyclic code of length $(\alpha,\beta)$ where $|\langle\theta\rangle|$ divides $\beta$ and $\gcd(\alpha, |\langle \Theta\rangle|)=1$. In Table \ref{tab2}, we provide generator polynomials $f(x)$ for $\mathcal{C}'$ and $g_1(x),g_2(x)$ for $\mathscr{C}=\xi_1\mathcal{C}_1\oplus \xi_2\mathcal{C}_2$ such that $x^{\alpha}-1=h(x)f(x)$, $x^{\beta}-1=h_i(x)g_i(x)$ for $i=1,2$. Also, we construct these codes under the conditions that $x^\alpha-1$ is divisible by $f(x)f^*(x)$ and $h^{\dagger}_i(x)h_i(x)$ is divisible by $x^\beta-1$ from the right side for $i=1,2$. Therefore, by Corollary \ref{coro-quantum}, we construct quantum codes (in the $5^{\text{th}}$ column), which have better parameters than the existing codes (in the $6^{\text{th}}$ column). To represent polynomials $f(x),g_i(x)$ for $i=1,2$, we write the vectors consisting of their coefficients
in decreasing order of power of $x$. For instance, we write the vector $1ww^72$ to represent the polynomial $x^3+wx+w^7x+2$.
\section{Conclusion}
In this paper, we have first discussed the structure of the ring $\mathbb{F}_q\mathcal{R}$ and studied the structural properties of $\mathbb{F}_q\mathcal{R}$-skew cyclic codes. As an application of our established results, we have constructed many quantum codes with better parameters using the CSS construction. To show the novelty of our obtained codes, we have compared them with the
best-known codes available in the literature. We believe that this work still has scope for further study by taking the direct product of cyclic, constacyclic, skew cyclic and skew constacyclic codes over finite rings.
\section*{Acknowledgements}
The first and second authors thank the Department of Science and Technology (DST), Govt. of India, for financial support under CRG/2020/005927, vide Diary No. SERB/F/67\\80/ 2020-2021 dated 31 December, 2020, and Ref No. DST/INSPIRE/03/2016/ 001445, respectively. H. Islam thanks the University of St. Gallen, Switzerland, for financial support under International Postdoctoral Fellowship (IPF). Also, the authors would like to thank the anonymous referee(s) and the Editor for their valuable comments to improve the presentation of the paper.	

\section*{Data Availability}
The authors declare that [the/all other] data supporting the findings of this study are available within the article. Any clarification may be requested from corresponding author provided it is essential. \\
\textbf{Competing interests}: The authors declare that there is no conflict of interest regarding the publication of this manuscript.\\

	\end{document}